\documentclass[a4paper,11pt]{article}
\usepackage[utf8]{inputenc}
\usepackage{amsmath,amssymb,amsthm}
\usepackage{graphicx}

\newtheorem{theorem}{Theorem}

\newtheorem{claim}{Claim}
\newtheorem{assert}{Assertion}
\newtheorem{conj}{Conjecture}
\newtheorem*{theoremvar}{Theorem 1'}

\title{Induced matchings in graphs of degree at most $4$}
\author{Viet Hang Nguyen \\ \footnotesize{EPFL,
Switzerland}\\\footnotesize{viethang.nguyen@epfl.ch}}
\date{}
\begin{document}

\maketitle

\begin{abstract}
We show that if $G$ is a connected graph of maximum degree at most $4$, which is
not $C_{2,5}$, then the strong matching number of $G$ is at least
$\frac{1}{9}n(G)$. This bound is tight and the proof implies a polynomial time
algorithm to find an induced matching of this size.
\end{abstract}

\section{Introduction}

An \emph{induced matching} of a graph $G=(V,E)$ is an edge set $M\subseteq E$
such that each vertex
of $G$ is incident to at most one edge in $M$ (i.e., $M$ is a matching) and if
$ab,cd \in M$ then none of the edges $ac, ad,bc,bd$ is in $E$. The
maximum size of an induced matching in $G$ is called the \emph{strong matching
number} of $G$ and is denoted by $\nu_s(G)$. 

While a maximum matching can be efficiently found in every graph
\cite{Edmonds65}, the problem of computing the strong matching number is
NP-hard even in quite restricted classes. It is proved to be NP-hard for
subcubic bipartite graphs \cite{SV82,Cameron89,Lozin02}, $C_4$-free
bipartite graphs
\cite{Lozin02}, line graphs \cite{KR03} or cubic planar graphs \cite{DWZ05}. In
fact, even for $3$-regular bipartite graphs, there is some constant $c>1$ such
that the prolem cannot be approximated within a factor of $c$ unless
$\text{P}=\text{NP}$ \cite{Zito99}.

On the positive side, a maximum induced matching can be found efficiently in
several classes of graphs such as weakly chordal graphs \cite{CST03}, AT-free
graphs \cite{Chang03}, graphs of bounded clique-width \cite{KR03}, and several
other classes \cite{BES07,BH08,Cameron04,GL93, GL00, Lozin02}.

One direction in recent research on induced matching is to lower bound the
strong matching number of a graph $G$ in terms of its maximum degree
$\Delta(G)$ and its order $n(G)$ or its number of edges $m(G)$. Let $G$ be a
connected graph, an easy observation \cite{Zito99} yields
\[
 \nu_s(G)\geq \frac{n(G)}{2(2\Delta(G)^2-2\Delta(G)+1)}.
\]
Joos \cite{Joos14} proved a sharp bound for $\Delta$ sufficiently large
\[
 \nu_s(G)\geq \frac{n(G)}{(\lceil \frac{\Delta}{2}\rceil+1)(\lfloor
\frac{\Delta}{2}\rfloor+1)}.
\]
He conjectured that this bound holds for all $\Delta\geq 3$ except for $G\in
\{C_{2,5},K_{3,3}^+\}$. Here $K_{3,3}^+$ is the graph obtained from $K_{3,3}$ by
subdividing an edge by a new vertex and $C_{2,5}$ is the graph obtained from
$C_5$
by separating each vertex into two non adjacent vertices (see
Figure \ref{fig:11}).
It is easy to see that $K_{3,3}^+$ is a subcubic graph,
$C_{2,5}$ is a $4$-regular graph and $\nu_s(K_{3,3}^+)=\nu_s(C_{2,5})=1$.

\begin{figure}[!h]
 \centering
 \includegraphics[width=.5\linewidth]{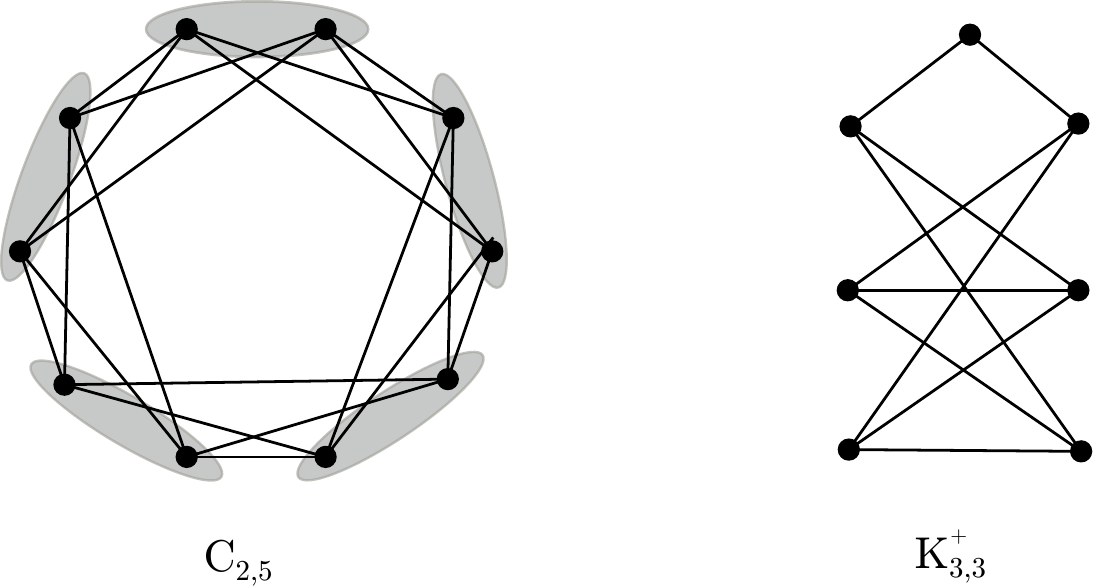}
 \caption{$C_{2,5}$ and $K^+_{3,3}$}
 \label{fig:11}
\end{figure}

For connected subcubic graphs, Joos, Rautenbach and Sasse \cite{Joos14} showed
that
$\nu_s(G)\geq \frac{n(G)}{6}$ if $G\neq K_{3,3}$. This result, proved by simple
local reduction, strengthens an earlier lower bound $\nu_s(G)\geq
\frac{1}{9}m(G)$ in \cite{KMM12} for subcubic planar graphs.

This research direction seems to be inspired by a conjecture of Erd\H{o}s and
Ne\v{s}et\v{r}il on the \emph{strong chromatic number} $\chi'_s(G)$, i.e., the
minimum number of induced matchings of $G$ into which $G$ can be partitioned. 

\begin{conj}
 \label{conj:EN}
 If $G$ is a connected graph with maximum degree $\Delta$ then
 \[
 \chi'_s(G)\leq \left\{ 
  \begin{array}{l l}
   \frac{5}{4}\Delta^2, &\text{$\Delta$ is even;}\\[.7em]
   \frac{1}{4}(5\Delta^2-2\Delta+1), &\text{$\Delta$ is odd}.
  \end{array}
\right.
 \]
\end{conj}

  The currently best known upper bound for the strong chromatic number is
$\chi'_s(G)\leq 1.998\Delta^2$ when $\Delta$ is sufficiently large, due to
Molloy and Reed \cite{MR97}. Conjecture \ref{conj:EN} is proved for subcubic
graphs in \cite{Andersen92,HQT93}.

The conjecture of Erd\H{o}s and Ne\v{s}et\v{r}il, if true, implies that for a
regular graph $G$ of even degree $\Delta$ we have $\nu_s(G)\geq
\frac{2n(G)}{5\Delta}$, as observed in \cite{JRS14}.
Note that Joos' conjectured bound strengthens this bound for $\Delta=4$ and
the result of Joos, Rautenbach and Sasse \cite{JRS14} confirms Joos' conjectured
bound for $\Delta=3$.

In this paper we prove the conjecture of Joos for $\Delta=4$, namely, we prove
the following.
\begin{theorem}
\label{thm:main}
 Let $G\neq C_{2,5}$ be a connected graph with maximum degree at most $4$.  Then
the strong matching number of $G$ is at least $\frac{1}{9}n(G)$. 
\end{theorem}

\section{Proof of the main theorem}
We first need some notations. For a subset $X$ of $V$ we denote by $G[X]$ the
subgraph of $G$ induced by the
vertices of $X$ and we use $G-X$ to denote $G[V-X]$. The number of isolated
vertices of a subgraph $H$ of $G$ is denoted by $i(H)$. For $X\subset V$ we
denote by
$d^{out}(X)$ the number of edges between $X$ and $V-X$. When $X=\{v\}$,
$d^{out}(X)$ is simply the degree of $v$ in $G$ and is written as $d(v)$. The
set of vertices adjacent to a vertex $v$ in $G$ is denoted by $N(v)$, noting
that
$|N(v)|=d(v)$. The minimum degree of a vertex in $G$ is denoted by $\delta(G)$.

In the remainder of this section we prove the following equivalent form of
Theorem \ref{thm:main}. 
\begin{theoremvar}
 Let $G$ be a graph with maximum degree at most $4$ and $G$ has no connected
component which is $C_{2,5}$. Then 
\[
 \nu_s(G)\geq \frac{1}{9}(n(G)-i(G)).
\]
\end{theoremvar}

 We proceed by contradiction. Suppose that $G$ is a counterexample to Theorem
1' of minimum order, namely, $\Delta(G)\geq 4$,
\begin{equation}
 \label{eq:nuG}
 \nu_s(G)<\frac{1}{9}n(G),
\end{equation}
and
\begin{equation}
\label{eq:nuG'}
 \text{\em for every proper subgraph $G'$ of $G$, $\nu_s(G')\geq
\frac{1}{9}((n(G')-i(G'))$.}
\end{equation}
It is worth remarking that the minimality of $G$ implies that $G$ is
connected and every subgraph of $G$ is not $C_{2,5}$.

The main point of our proof is to show that $G$ satisfies
\begin{equation}
\label{eq:1}
\text{\em{$G$ has girth at least $5$ and $\delta(G)\geq 3$.}}
\end{equation}

First, let us see how (\ref{eq:1}) implies contradiction. Suppose that $G$
satisfies (\ref{eq:1}). Let $uv$ be any edge of $G$. Then $G$ is connected by
the assumption on
minimality. Let $X=N(u)\cup N(v)$. Then $|X|\leq 8$ as the maximum degree of
$G$ is at most $4$. Since the girth
of $G$ is at least $5$, each vertex in $V(G)-X$ is adjacent to at most $2$
vertices in $X$. Combining this with the assumption $\delta(G)\geq 3$ we obtain
that there is no isolated vertex
in $G'=G-X$. Moreover, $G'\neq C_{2,5}$ as remarked above. Therefore,
$\nu_s(G')\geq\frac{1}{9}n(G')$ holds. It is easy to see
that if $M$ is an induced matching of $G'$ then $M\cup \{uv\}$ is an induced
matching of $G$. Therefore,
\[
 \nu_s(G)\geq 1=\nu_s(G')+1\geq \frac{1}{9}n(G')+1\geq
\frac{1}{9}(n(G)-8)>\frac{1}{9}n(G)
\]
holds, contradicting (\ref{eq:nuG}).

Next, we will prove (\ref{eq:1}) through a sequence of claims using local
reduction, similar to \cite{JRS14}. 
We call a vertex of degree $1$ in $G$ an \emph{end-vertex}. For an induced
subgraph $G'$ of $G$ we denote by $I(G')$ the set of isolated vertices in $G'$
and $I_j(G')\subseteq I(G')$ the set of isolated vertices in $G'$ which has
degree $j$ in $G$. The cardinalities of $I(G'),I_j(G')$ are denoted by
$i(G'),i_j(G')$, for $j\in \{1,2,3,4\}$, respectively. Then we have
\begin{equation}
 \label{eq:iG'}
 i(G')=i_1(G')+i_2(G')+i_3(G')+i_4(G').
\end{equation}

For $X\subset V$, if $G'=G-X$ then, since each vertex in $I(G')$ is
adjacent only to vertices in $X$, one can see that
\begin{equation}
\label{eq:dout}
 d^{out}(X)\geq i_1(X)+2i_2(X)+3i_3(X)+4i_4(X).
\end{equation}
We also have
\begin{equation}
 \label{eq:dout2}
 d^{out}(X\cup I(G'))=d^{out}(X)-(i_1(X)+2i_2(X)+3i_3(X)+4i_4(X)).
\end{equation}

\begin{claim}
 \label{clm:1}
 The neighbor of an end-vertex has degree $4$.
\end{claim}
\begin{proof}
Let $u$ be an end-vertex and $v$ its unique neighbor. 
Suppose to the contrary that $d(v)\leq 3$. Let $X=\{v\}\cup N(v)$ and $G'=G-X$.
Then a simple counting shows that $|X|\leq 4$ and
$d^{out}(X)\leq 6$, thus $i(G')\leq 6$, by (\ref{eq:iG'}) and (\ref{eq:dout}).

If both
$i(G')=6$ and $|X|=4$ hold, $G$ is the graph in Figure \ref{fig:5} and
it is easy to see that $\nu_s(G)=2>\frac{1}{9}n(G)$, contradicting
(\ref{eq:nuG}). 

If $i(G')< 6$ or
$|X|<4$, then, noting that $n(G')=n(G)-|X|$, we have $\nu_s(G)\geq
1+\nu_s(G')\geq
1+\frac{1}{9}(n(G')-i(G'))\geq 1+\frac{1}{9}(n(G)-9)=\frac{1}{9}n(G)$, again a
contradiction to (\ref{eq:nuG}). 
\end{proof}

\begin{figure}[h!]
\centering
 \includegraphics[width=.2\linewidth]{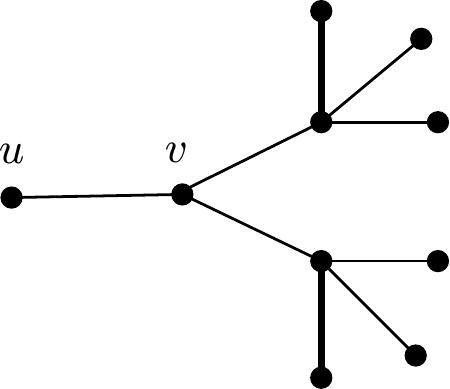}
 \caption{The graph in the proof of Claim \ref{clm:1}; thick edges indicate the
induced matching.}
 \label{fig:5}
\end{figure}

\begin{claim}
 \label{clm:2}
 No two end-vertices have a common neighbor.
\end{claim}
\begin{proof}
Suppose to the contrary that two end-vertices $u_1,u_2$ have a common neighbor
$v$.
Then $d(v)=4$ by Claim \ref{clm:1}. Let $X=\{v\}\cup N(v)=\{v,u_1,u_2,w_1,w_2\}$
and $G'=G-X$.
Then $|X|=5$ and $d^{out}(X)\leq 6$. 

If $i(G')\leq 4$ then $\nu_s(G)\geq
1+\nu_s(G')\geq 1+\frac{1}{9}(n(G')-i(G'))\geq \frac{1}{9}n(G)$, a
contradiction. Therefore, let
us suppose that $i(G')\geq 5$. Then both $w_1,w_2$ must be adjacent to some end-
vertices, say $t_1, t_2$, and moreover $w_1$ and $w_2$ are not adjacent,
otherwise $d^{out}(X)\leq 4$, which implies $i(G')\leq 4$, a contradiction.

Let $X'=X\cup N(w_1)\cup N(w_2)$ and $G''=G-X'$. Then
$|X'|\leq 11$ and $d^{out}(X')\leq 3$, which implies $i(G'')\leq 3$ (see Figure
\ref{fig:6}). Since for
every induced matching $M$ of $G''$, $M\cup\{w_1t_1,w_2t_2\}$ is an induced
matching of $G$, we have 
\[
\nu_s(G)\geq 2+\nu_s(G'')\geq
2+\frac{1}{9}(n(G'')-i(G''))\geq 2+\frac{1}{9}(n(G)-11-3)>\frac{1}{9}n(G),
\]
a contradiction.
\end{proof}

\begin{figure}[h!]
\centering
 \includegraphics[width=.3\linewidth]{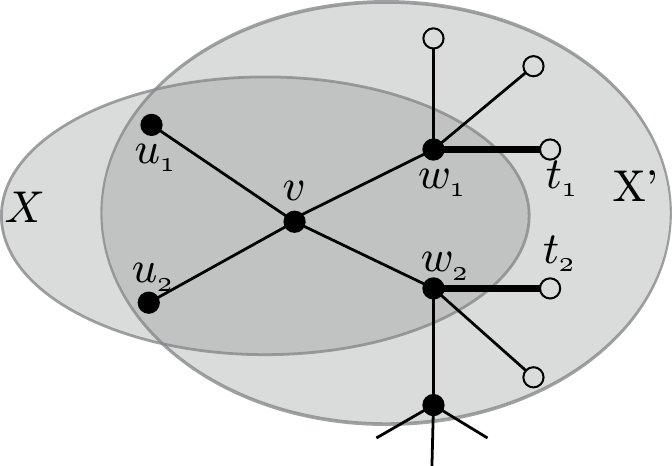}
  \caption{An illustration for the proof of Claim \ref{clm:2}}
  \label{fig:6}
\end{figure}

\begin{claim}
 \label{clm:3}
 No two end-vertices have distance $4$ in $G$.
\end{claim}
\begin{proof}

Suppose to the contrary that two end-vertices $u_1$ and $u_2$ have distance $4$
in $G$ and $u_1v_1wv_2u_2$ is a
path of length $4$ linking $u_1$ and $u_2$. 
Consider $X=\{v_1,v_2\}\cup N(v_1)\cup N(v_2)$. Then $|X|\leq 9$,
and $d^{out}(X)\leq 14$.
From (\ref{eq:iG'}) and (\ref{eq:dout}) we have $d^{out}(X)\geq
2i(G')-i_1(G')$. Thus, if $i(G')\geq 10$ then $i_1(G')\geq
2i(G')-d^{out}(X)\geq 6$. However, vertices in $I_1(G')$ are only
adjacent to vertices in set $X-\{v_1,v_2,u_1,v_1\}$, which consists of at most
$5$ vertices. Hence, there are two end-vertices in $I_1(G')$ that have a common
neighbor, contradicting Claim \ref{clm:2}. Therefore, $i(G')\leq 9$, and, noting
that each induced matching $M$ of $G'$ can be extended to an induced matching
$M\cup \{u_1v_1,u_2v_2\}$ of $G$, we derive that
\[
\nu_s(G)\geq 2+\nu_s(G')\geq 2+\frac{1}{9}(n(G')-i(G'))\geq 2+
\frac{1}{9}(n(G)-9-9)=\frac{1}{9}n(G),
\]
a contradiction to (\ref{eq:nuG}).
\noindent(See Figure \ref{fig:7}.)
\end{proof}

\begin{figure}[!h]
\centering
 \includegraphics[width=.4\linewidth]{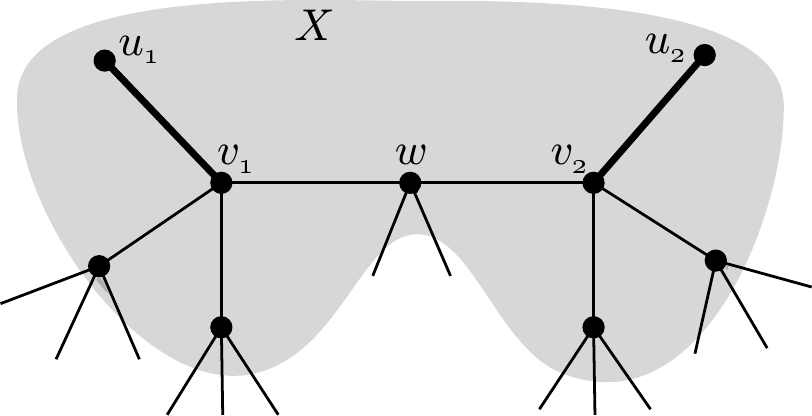}
 \caption{An illustration for the proof of Claim \ref{clm:3}}
 \label{fig:7}
\end{figure}

\begin{claim}
 \label{clm:4}
 $\delta(G)\geq 2$.
\end{claim}
\begin{proof}
Suppose to the contrary that $u$ is a vertex of degree $1$ and $v$ is its unique
neighbor in $G$. Then, by Claim \ref{clm:1}, $v$ has $3$ other neighbors
$w_1,w_2,w_3$. Let $X=\{u,v,w_1,w_2,w_3\}$ and
$G'=G-X$. Since each induced matching $M$ of $G'$ can be extended to an induced
matching $M\cup \{uv\}$ of $G$, if $i(G')\leq 4$ then 
\[
 \nu_s(G)\geq 1+\nu_s(G')\geq 1 + \frac{1}{9}(n(G')-i(G'))\geq
1+\frac{1}{9}(n(G)-5-4)=\frac{1}{9}n(G)
\]
holds, contradicting (\ref{eq:nuG}). Therefore, let us suppose that
$i(G')\geq 5$.
Since
\[
 2i(G')-i_1(G')\leq
d^{out}(X)\leq 9,
\]
 we must have $i_1(G')\geq 1$. By Claim \ref{clm:2} and \ref{clm:3}, if $s$ and
$t$ are
two vertices in $I_1(G')$ then they are adjacent to two distinct vertices $w_i$
and $w_j$ and furthermore, $w_i$ and $w_j$ are adjacent. Hence,
\[
 d^{out}(X)\leq 9-2(i_1(G')-1)=11-2i_1(G').
\]

Therefore, by (\ref{eq:dout}), $11\geq 3i_1(G')+2i_2(G')+3i_3(g')+4i_4(G')$
holds, which, together
with $i_1(G')\geq 1$ and $i(G')\geq 5$, implies that $i_1(G')=1,
i_2(G')=4$, $i_3(G')=i_4(G')=0$ and furthermore $w_1,w_2,w_3$ all have degree
$4$ and are not adjacent. But then $d^{out}(X\cup I(G'))=0$ and hence $n(G)=10$.
Let $s$ be
the unique vertex in $I_1(G')$ and suppose without loss of generality that its
unique neighbor in $G$ is $w_1$. Since $d(w_1)=4$ and $i_2(G')=4$, there is a
vertex $t$ in $I_2(G')$ that is not adjacent to $w_1$. Then, $\{w_1s, w_2t\}$
is an induced matching of $G$, and thus $\nu_s(G)\geq 2> \frac{1}{9}n(G)$, a
contradiction. 
\end{proof}

\begin{claim}
 \label{clm:5}
 No two vertices of degree $2$ are adjacent.
\end{claim}
\begin{proof}
Suppose to the contrary that $u,v$ are two adjacent vertices of degree $2$.
Let $X=N(u)\cup N(v)$ and $G'=G-X$. Then $|X|= 4$ and $d^{out}(X)\leq 6$.
Remembering that
$i_1(G)=0$ by Claim \ref{clm:4}, we see that $i(G')\leq
\lfloor\frac{1}{2}d^{out}(X)\rfloor\leq 3$ by (\ref{eq:dout}). Since for each
induced matching $M$ of $G'$, $M\cup\{uv\}$ is an induced matching of $G$, it
follows that 
\[
 \nu_s(G)\geq
1+\nu_s(G')\geq 1+ \frac{1}{9}(n(G')-i(G'))\geq 1+\frac{1}{9}(n(G)-4-3)\geq
\frac{1}{9}n(G)
\]
holds, a contradiction.
\end{proof}

\begin{claim}
 \label{clm:6}
 No vertex of degree $2$ is contained in a triangle.
\end{claim}

\begin{proof}
Suppose to the contrary that $u$ is a vertex of degree $2$ which is contained in
a triangle and
$v$ is one of its neighbors. Let $X=N(u)\cup N(v)$ and $G'=G-X$. Then $|X|\leq
5$ and $d^{out}(X)\leq 8$. On the other hand, by Claim \ref{clm:4} and
(\ref{eq:dout}),
$d^{out}(X)\geq 2i_2(G')+3i_3(G')+4i_4(G')$ and $i(G')=i_2(G')+i_3(G')+i_4(G')$
hold. 
Hence $i(G')\leq d^{out}(X)/2\leq 4$. Therefore,
\[
 \nu_s(G)\geq 1+\nu_s(G')\geq 1+\frac{1}{9}(n(G)-5-4)=\frac{1}{9}n(G),
\]
a contradiction.
\end{proof}

\begin{claim}
 \label{clm:7}
 No vertex of degree $2$ is contained in a cycle of length $4$.
\end{claim}
\begin{proof}
Suppose that $u$ is a vertex of degree $2$ which is contained in a cycle
$uvwt$ of
length $4$. Let $X=N(u)\cup N(v)$ and $G'=G-X$. Then $|X|\leq 6$ and
$d^{out}(X)\leq 10$, thus $i(G')\leq d^{out}(X)/2 \leq
5$.

If $i(G')\leq 3$ then $\nu_s(G)\geq 1+\nu_s(G')\geq \frac{1}{9}n(G)$,
contradicting (\ref{eq:nuG}). Hence we may suppose that $i(G')\geq 4$. Then
there exists a vertex $s$ in $I(G')$
that is not adjacent to $t$. Since $\delta(s)\geq 2$ by Claim \ref{clm:4}, $s$
is adjacent to a vertex $r$ in $N(v)-\{u,w\}$. Also we have
$|N(X\cup I(G'))|\leq d^{out}(X\cup I(G'))\leq d^{out}(X)-2i(G')\leq 2$, where
the second inequality follows from (\ref{eq:iG'}), (\ref{eq:dout}),
(\ref{eq:dout2}) and Claim \ref{clm:4}.

Now let $X'=X\cup I(G')\cup N(X\cup I(G'))$ and $G''=G-X'$. Then $|X'|\leq
|X|+i(G')+|N(X\cup I(G'))|\leq 6+5+2 = 13$.
Moreover, $d^{out}(X')\leq 3|N(X\cup I(G'))|\leq 6$. Hence $i(G'')\leq
d^{out}(X')/2\leq 3$. Since for each induced matching $M$ of $G''$,
$M\cup\{ut,rs\}$ is an induced matching of $G''$, it follows that
\[
 \nu_s(G)\geq 2+\nu_s(G'')\geq \frac{1}{9}n(G),
\]
contradicting (\ref{eq:nuG}).
(see Figure \ref{fig:9})
\end{proof}
\begin{figure}
\centering
 \includegraphics[width=.35\linewidth]{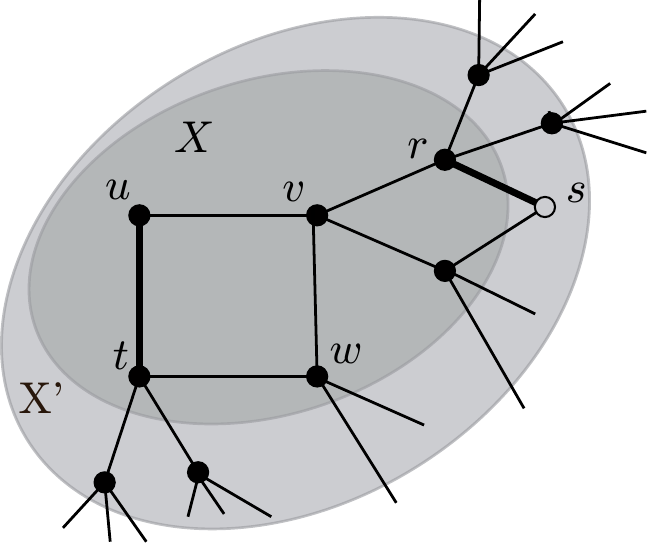}
 \caption{An illustration for the proof of Claim \ref{clm:7}.}
 \label{fig:9}
\end{figure}

\begin{claim}
 \label{clm:8}
 $\delta(G)\geq 3$.
\end{claim}
\begin{proof}
Suppose to the contrary that $u$ is a vertex of degree $2$ in $G$ and $v,w$ are
its neighbors.
Let $X=N(u)\cup N(v)$ and $G'=G-X$. Then $|X|\leq 6$ and $d^{out}(X)\leq 12$.
By Claim \ref{clm:7}, no vertex of degree $2$ is contained in a cycle of length
$4$, thus $w$ is not adjacent to any neighbor of $v$ other than $u$. Now if
$i(G')\leq 3$ then $\nu_s(G)\geq 1+\nu_s(G')\geq
1+\frac{1}{9}(n(G)-6-3)=\frac{1}{9}n(G)$, a contradiction. 

So let us suppose that
$i(G')\geq 4$. Then there is a vertex $s\in I(G')$ that is not adjacent to $w$,
note that then $s$ is adjacent only to vertices in $N(v)-\{u\}$, so $s$ is
contained in a cycle of length $4$. Therefore $d(s)\geq 3$ by Claim \ref{clm:7}.
Let
$t\in N(v)$ be one of its
neighbors. Then $\{st, uw\}$ is an induced matching of $G$ (see Figure
\ref{fig:10}). 

Since vertices in $I(G')$ are adjacent only to vertices
in $X$ we have $d^{out}(X\cup I(G'))\leq d^{out}(X)-\sum_{x\in I(G')}d(x)$.
Since
$d(x)\geq 2$ for all $x\in I(G')$, by Claim \ref{clm:4}, $d(s)\geq 3$ and
$i(G')\geq 4$, we have
$\sum_{x\in I(G')}d(x)\geq3\times 2+3=9$. Thus $d^{out}(X\cup I(G'))\leq 3$.

Let $X'=X\cup I(G')\cup N(X\cup I(G'))$ and $G''=G-X'$. Then one can
easily see
that $d^{out}(X')\leq 3|N(X\cup I(G'))|\leq 3d^{out}(X\cup I(G'))\leq 9$, so
$i(G'')\leq \lfloor
\frac{d^{out}(X')}{2}\rfloor =4$, by Claim \ref{clm:4} and (\ref{eq:dout}). Now
let us upper bound $|X'|$. We have \begin{eqnarray*}
 |X'|&\leq& |X|+i(G')+d^{out}(X\cup I(G')\\
 &\leq &|X|+i(G')+d^{out}(X)-2i(G')\\
 &=&|X|+d^{out}(X)-i(G')\\
 &\leq&6+12-4\\
 &=&14.
\end{eqnarray*}

Noting that every matching $M$ of $G''$ can
be extended to a matching $M\cup \{st,uw\}$ of $G$, we derive
\[
 \nu_s(G)\geq 2+\nu_s(G'')\geq 2+\frac{1}{9}(n(G)-14-4)=\frac{1}{9}n(G),
\]
a contradiction.
\end{proof}
\begin{figure}[!h]
\centering
 \includegraphics[width=.35\linewidth]{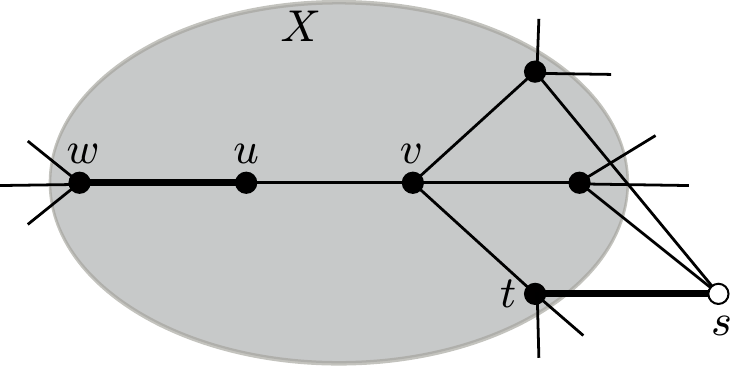}
 \caption{An illustration for the proof of Claim \ref{clm:8}.}
 \label{fig:10}
\end{figure}

\begin{claim}
 \label{clm:9}
 $G$ contains no triangle.
\end{claim}
\begin{proof}
Suppose to the contrary that $uvw$ is a triangle in $G$. Let $X=N(u)\cup N(v)$
and $G'=G-X$.
Then $|X|\leq 7$. So if $i(G')\leq 2$ we have $\nu_s(G)\geq 1+\nu_s(G')\geq
1+\frac{1}{9}(n(G)-7-2)=\frac{1}{9}n(G)$, a contradiction. 

Let us suppose
that $i(G')\geq 3$. It is easy to see that $d^{out}(X)\leq 14$. Hence, by
Claim \ref{clm:8} and (\ref{eq:dout}), 
$i(G')=i_3(G')+i_4(G')\leq \lfloor d^{out}(X)/3\rfloor \leq 4$.
Since $|N(w)-X|\leq 2$, and $\delta(G)\geq 3$ by Claim \ref{clm:8}, there exists
a
vertex $s\in I(G')$ such that $s$ is not
adjacent to $w$. Then $s$ is adjacent to a vertex $r\in N(v)-\{u,w\}$ and $\{uw,
sr\}$ is an induced matching of $G$. 

Let $X'=X\cup I(G')\cup N(w)\cup N(r)$ and
$G''=G-X'$. Let $Y=(N(w)\cup N(r))- (X\cup I(G')$. Then by simple counting one
can see that $|Y|\leq 4$ and $|X'|=|X|+ i(G')+|Y|\leq 11+i(G')$.

Since each vertex in $Y$ is adjacent to at least one vertex in $X$ and to at
most $3$ vertices outside $X$, we have
\begin{eqnarray*}
d^{out}(X')&\leq& d^{out}(X\cup I(G')) -|Y|+3|Y|\\
&\leq& d^{out}(X)-3 i(G')+2|Y|,
\end{eqnarray*}
where the last equality follows from \ref{eq:dout2} and Claim \ref{clm:8}.
Therefore, 
\begin{eqnarray*}
|X'|+i(G'')&\leq& 11+i(G')+\Big{\lfloor} \frac{d^{out}(X')}{3} \Big{\rfloor}\\
&\leq& 11+\Big{\lfloor} \frac{d^{out}(X)+2|Y|}{3}\Big{\rfloor}\\
&\leq& 11 +\Big{\lfloor} \frac{14+2\times 4}{3}\Big{\rfloor}\\
&=&18.
\end{eqnarray*}
Since each induced matching $M$ of $G''$ can be extended to an induced
matching $M\cup \{uw,sr\}$ of $G$, it follows that
\[\nu_s(G)\geq 2+\nu_s(G'')\geq 2+\frac{1}{9}(n(G)-18)=\frac{1}{9}n(G),\]
holds, contradicting (\ref{eq:nuG}).

\end{proof}

\begin{claim}
 \label{clm:10}
 If a vertex $u$ is contained in a cycle of length $4$ then $d(u)\geq 4$.
\end{claim}
\begin{proof}
Suppose to the contrary that $u$ is contained in a cycle $C=uvxy$ of length
$4$ and $d(u)=3$.
Let $v$ be the neighbor of $u$ that is not contained in $C$. Let $X=N(u)\cup
N(v)$
and $G'=G-X$. Then $|X|\leq 7$ and $d^{out}(X)\leq 13$. Hence, if $i(G')\leq 2$
then $\nu_s(G)\geq 1+\nu_s(G')\geq \frac{1}{9}n(G)$ holds, a contradiction. So
let us suppose that $i(G')\geq 3$. 

If $d(v)=3$ then $|X|=6$ and $d^{out}(X)\leq 10$, thus $i(G')\leq 3$. Hence
$\nu_s(G)\geq 1+\nu_s(G')\geq \frac{1}{9}n(G)$ holds, again a contradiction. So
we may suppose that
$d(v)=4$ and $N(v)=\{u,x,t,r\}$.

We will use the following assertion.

\begin{assert}
 \label{as:1}
 If there is an induced matching in $G[X\cup I(G')]$ then $\nu_s(G)\geq
\frac{1}{9}n(G)$.
\end{assert}

\begin{proof}
Let $\{ab,cd\}$ be an induced matching in $G[X\cup I(G')]$. Let $X_1=X\cup
I(G')\cup N(\{a,b,c,d\})$ and $Y_1=N(\{a,b,c,d\})-(X\cup I(G')$. Then
since each vertex in $Y_1$ is adjacent to
a vertex in $\{a,b,c,d\}\subset X\cup I(G')$ and $\delta(G)\geq 3$ by Claim
\ref{clm:8}, we have $|Y_1|\leq d^{out}(X\cup
I(G'))-3i(G')\leq 13- 3\times 3=4$. Using a similar argument as in the proof of
Claim \ref{clm:9}, we obtain 
\[
 d^{out}(X_1)\leq d^{out}(X)-3i(G')-|Y_1|+3|Y_1|=d^{out}(X)-3i(G')+2|Y_1|,
\]
and 
\[
 |X_1|=|X|+|I(G')|+|Y_1|.
\]
Therefore,
\begin{eqnarray*}
 |X_1|+i(G_1)&\leq&
|X|+i(G')+|Y_1|+\frac{d^{out}(X_1)}{3}\\
&=&|X|+\frac{5}{3}|Y_1|+\frac{d^{out}(X)}{3}\\
&\leq&7+\frac{5}{3}\times 4+\frac{13}{3}\\
&=&18.
\end{eqnarray*}
Hence, since for each induced matching $M$ of $G''$, $M\cup\{ab,cd\}$ is an
induced matching of $G$, we have
\[
 \nu_s(G)\geq 2+\nu_s(G_1)\geq 2+\frac{1}{9}(n(G)-18)=\frac{1}{9}n(G),
\]
a contradiction.
\end{proof}

Now we are ready to complete the proof of Claim \ref{clm:10}. Since $i(G')\geq
3$ and $|N(t)-X|, |N(r)-X|\leq 2$, by observing
that each vertex in $I(G')$ must be adjacent to either $r$ or $t$, we deduce
that there are two vertices $s_1,s_2\in I(G')$ such that $s_1$ is adjacent to
$r$ but not to $t$ and $s_2$ is adjacent to $t$ but not to $r$. Then since $r$
and $t$ are not adjacent by Claim \ref{clm:9}, $\{s_1r,s_2t\}$ is an induced
matching of $G[X\cup I(G')]$. So Assertion \ref{as:1} implies contradiction.

\end{proof}
\begin{claim}
 \label{clm:11}
 There is no cycle of length $4$.
\end{claim}

\begin{proof}
Suppose that $C=uvxy$ is a cycle of length $4$ in $G$. Let $X=N(u)\cup N(v)$
and $G'=G-X$. Then $|X|\leq 8$ and $d^{out}(X)\leq 16$. If $i(G')\leq 1$ then
a similar argument as in previous claims yields $\nu_s(G)\geq \frac{1}{9}n(G)$,
a contradiction. Therefore, we may suppose that $i(G')\geq 2$. 

Since there are no three vertices in $X$ with pairwise distance $3$ and each
vertex in $I(G')$ has degree at least $3$ we obtain that each vertex of $I(G')$
lies on a
cycle of length $4$. Therefore, by Claim \ref{clm:10} we have
\begin{equation}
\label{eq:2}
 \text{all vertices in $I(G')$ have degree $4$.}
\end{equation}

We will use the following assertion, which is similar to Assertion \ref{as:1}.

\begin{assert}
 If there is an induced matching in $G[X\cup I(G')]$ then
$\nu_s(G)\geq \frac{1}{9}n(G)$.
\end{assert}
\begin{proof}
Let $\{ab,cd\}$ be an induced matching in $G[X\cup I(G')]$. Let $X_1=X\cup
I(G')\cup N(\{a,b,c,d\})$ and $Y_1=N(\{a,b,c,d\})-(X\cup I(G'))$. 

We first prove that if $|Y_1|\leq 4$ then $\nu_s(G)\geq \frac{1}{9}n(G)$, a
contradiction. In fact, a similar argument as in the proof of Assertion
\ref{as:1} yields 
\[
 d^{out}(X_1)\leq d^{out}(X)-4i(G')-|Y_1|+3|Y_1|=d^{out}(X)-4i(G')+2|Y_1|,
\]
 where the multiplicity $4$ for $i(G')$ is due to (\ref{eq:2}).  We also
have 
\[
 |X_1|=|X|+|I(G')|+|Y_1|.
\]
Therefore,
\begin{eqnarray*}
 |X_1|+i(G_1)&\leq&
|X|+i(G')+|Y|+\frac{d^{out}(X_1)}{4}\\
&=&|X|+\frac{3}{2}|Y_1|+\frac{d^{out}(X)}{4}\\
&\leq&8+\frac{3}{2}\times 4+\frac{16}{4}\\
&=&18.
\end{eqnarray*}
It follows that 
\[
 \nu_s(G)\geq 2+\nu_s(G_1)\geq 2+\frac{1}{9}(n(G)-18)=\frac{1}{9}n(G)
\]
holds.

It remains to prove that $|Y_1|\leq 4$. Indeed, if
$|\{a,b,c,d\}\cap I(G')|\geq  2$, then it is easy to see that
$|Y_1|=|N(\{a,b,c,d\}-(X\cup I(G'))|\leq 4$. On the other hand, if
$\{a,b,c,d\}\subset X$ then a simple counting shows that $d^{out}(X)\leq 12$
and hence $|Y_1|\leq d^{out}(X)-4i(G')\leq 4$, where the coefficient $4$ of
$i(G')$ is due to (\ref{eq:2}). Thus we may suppose that $a\in
I(G')$ and   $\{b,c,d\}\subset X$, note that then $b\notin \{u,v\}$. Let $t\in
I(G')-s$. First consider the case $\{c,d\}\cap \{u,v\}=\emptyset$. 
Since $d(s)=d(t)=4$ and $|X-\{u,v,b,c,d\}|=3$ we have that both $s$ and $t$ have
a neighbor in $\{b,c,d\}$. This implies that $|Y_1| =|N(\{a,b,c,d\})-(X\cup
I(G')|\leq 4$, as desired. Next, consider the case $\{c,d\}= \{u,v\}$, it is
easy to see that then $|Y_1|\leq 2$. Finally, consider the case $|\{c,d\}\cap
\{u,v\}|=1$, say $d\in \{u,v\}$ and $c\notin \{u,v\}$. Then since
$|N(t)\cap\{x,y\}|\leq 1$, as $G$ contains no triangle, $d(t)=4$ and
$|X-\{u,v,x,y\}|=4$, we have that $t$ must be adjacent to $b$ or $c$. Hence
again we obtain $|Y-1|\leq 4$, as desired. 
\end{proof}

Now suppose that there is no induced matching of size $2$ in $G[X\cup I(G')]$.
We prove that $G=C_{2,5}$. Let $N(v)=\{u,x,a,b\}$ and $N(u)=\{v,y,c,d\}$. Since
each vertex $s$ in $I(G')$ is not adjacent to both $x$ and $y$ at the same
time as $G$ has no triangle, we may assume without loss of generality that
there is an
$s\in I(G')$ that is adjacent to $a,c$ but not to $y$. Then since $\{sa,yu\}$ is
not an induced matching of $G[X\cup I(G')]$ and since $ua\notin E$ by Claim
\ref{clm:9}, we must have $ya\in E$. (See Figure \ref{fig:1} for an
illustration.) Now by considering pair $\{du,sa\}$, since $ua,us\notin E$, we
see that either $da\in E$ or $sd\in E$, but not both of them belong to $E$.

\begin{figure}[!h]
\centering
 \includegraphics[width=.6\linewidth]{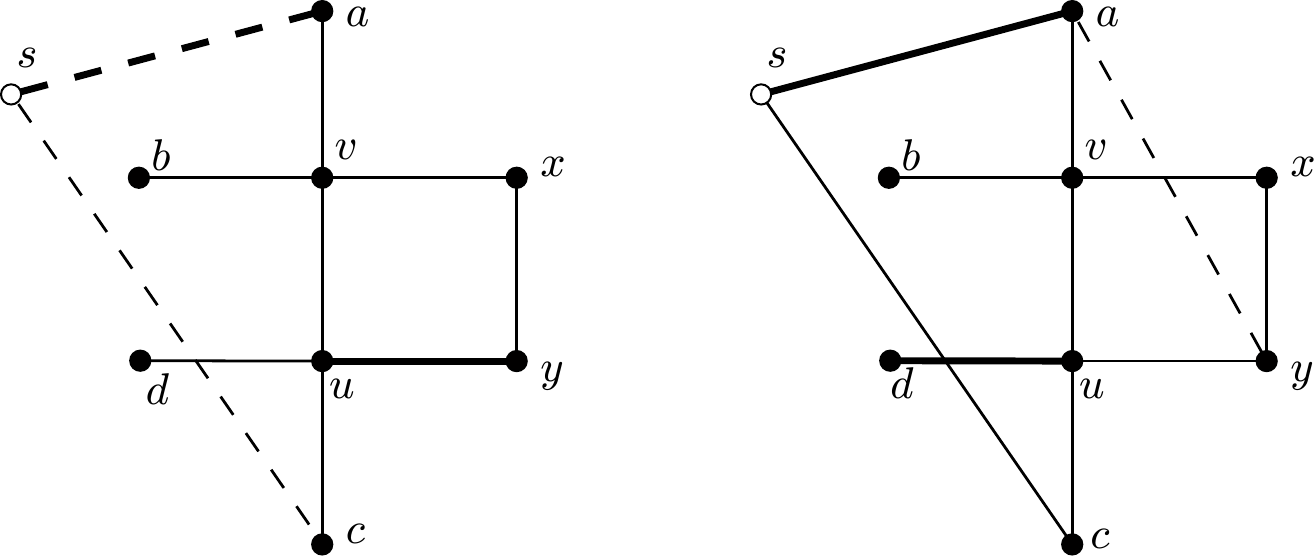}
 \caption{An illustration of Case 1; dashed edges indicate newly
``found'' edges and thick edges indicate considered pair in each step.}
 \label{fig:1}
\end{figure}
\medskip

\noindent \underline{Case 1}: {\em $da\in E$ and $sd\notin E$}.\\
\indent \indent (See an illustration for this case in Figure \ref{fig:2})

Since $d(s)=4$ we derive that $sx, sb\in E$. Consider pair $\{sx,ud\}$, we see
that $dx\in E$ since $ux,us,ud\notin E$. 
  Now let $t$ be a vertex in $I(G')-\{s\}\neq \emptyset$, due to
the assumption that $I(G')\geq 2$. Then since $d(a)=d(x)=4$, we know that
$N(a)=\{s,t,x,y\}$, $N(x)=\{b,d,y,v\}$ and $ta\notin E$. Therefore,
$N(t)\subseteq \{b,c,d,x,y\}$ holds.  However, by considering pair $\{sb,ud\}$
we derive that $bd\in E$ which implies that $G$ contains a triangle $tbd$,
contradicting Claim \ref{clm:9}.
\begin{figure}
\centering
 \includegraphics[width=.9\linewidth]{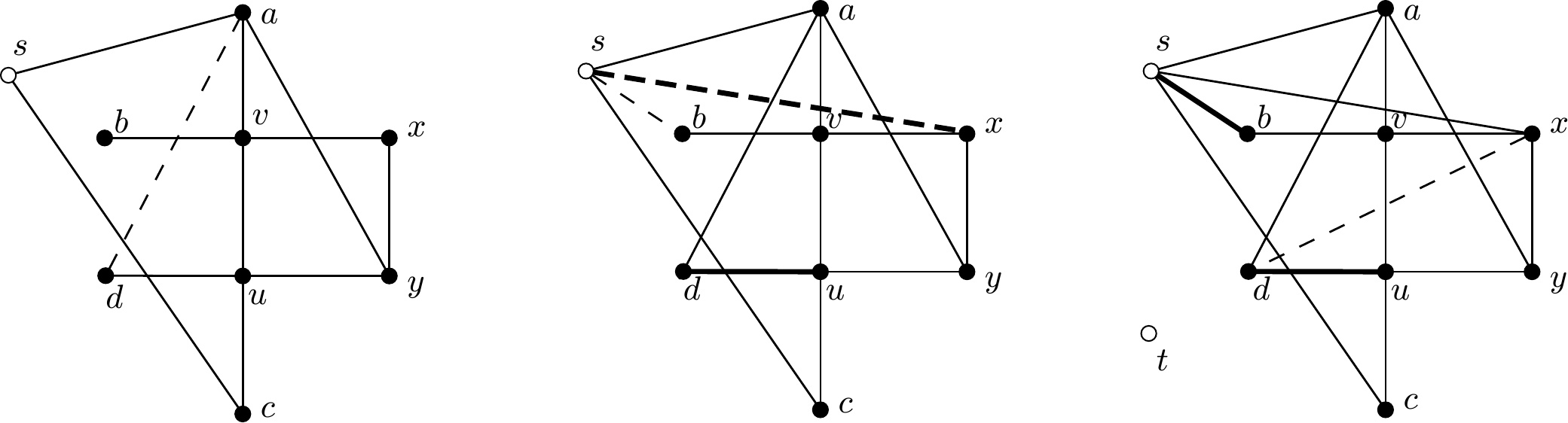}
 \caption{An illustration for Case 1; dashed edges indicate newly
``found'' edge and thick edges indicate considered pair in each step.}
 \label{fig:2}
\end{figure}
\medskip

\noindent \underline{Case 2}: {\em $sd\in E$ and $da\notin E$}. 

Consider pair $\{sc,bv\}$, using the fact that $G$ contains no triangle
from Claim \ref{clm:9}, we see that either $bs$ or $bc$ is in $E$ but not both
of them. Hence we consider these cases separately.\medskip

\noindent \underline{Case 2.1}: {\em $bc\in E$ and $bs\notin E$}. \\
\indent \indent(See Figure \ref{fig:3} for an illustration).

Then, since $d(s)=4$ and $sy\notin E$, we have that $sx\in E$. Also we have
$ac\notin E$, otherwise $asc$ is a triangle in $G$, contradicting Claim
\ref{clm:9}. By considering pair $\{bc,ay\}$ we obtain that $by\in E$. Now let
$t \in I(G')-\{s\}\neq \emptyset$. Since $N(y)=\{x,u,a,b\}$ and $tb,tc$ do
not both belong to $E$, we must have $tx,td,ta\in E$. Considering pair
$\{uc,tx\}$,
noting that $xc\notin E$ since otherwise $scx$ is a triangle in $G$, we derive
that $tc\in E$. Finally, by considering pair $\{by,sd\}$, we obtain that $bd\in
E$.
However then $V(G)=X$ and $G=C_{2,5}$, where the pairs corresponding to the
vertices of $C_5$ are $\{a,x\},\{v,y\},\{b,u\},\{c,d\},\{s,t\}$, a
contradiction to the assumption that $G\neq C_{2,5}$. 

\begin{figure}[!h]
 \centering
 \includegraphics[width=.9\linewidth]{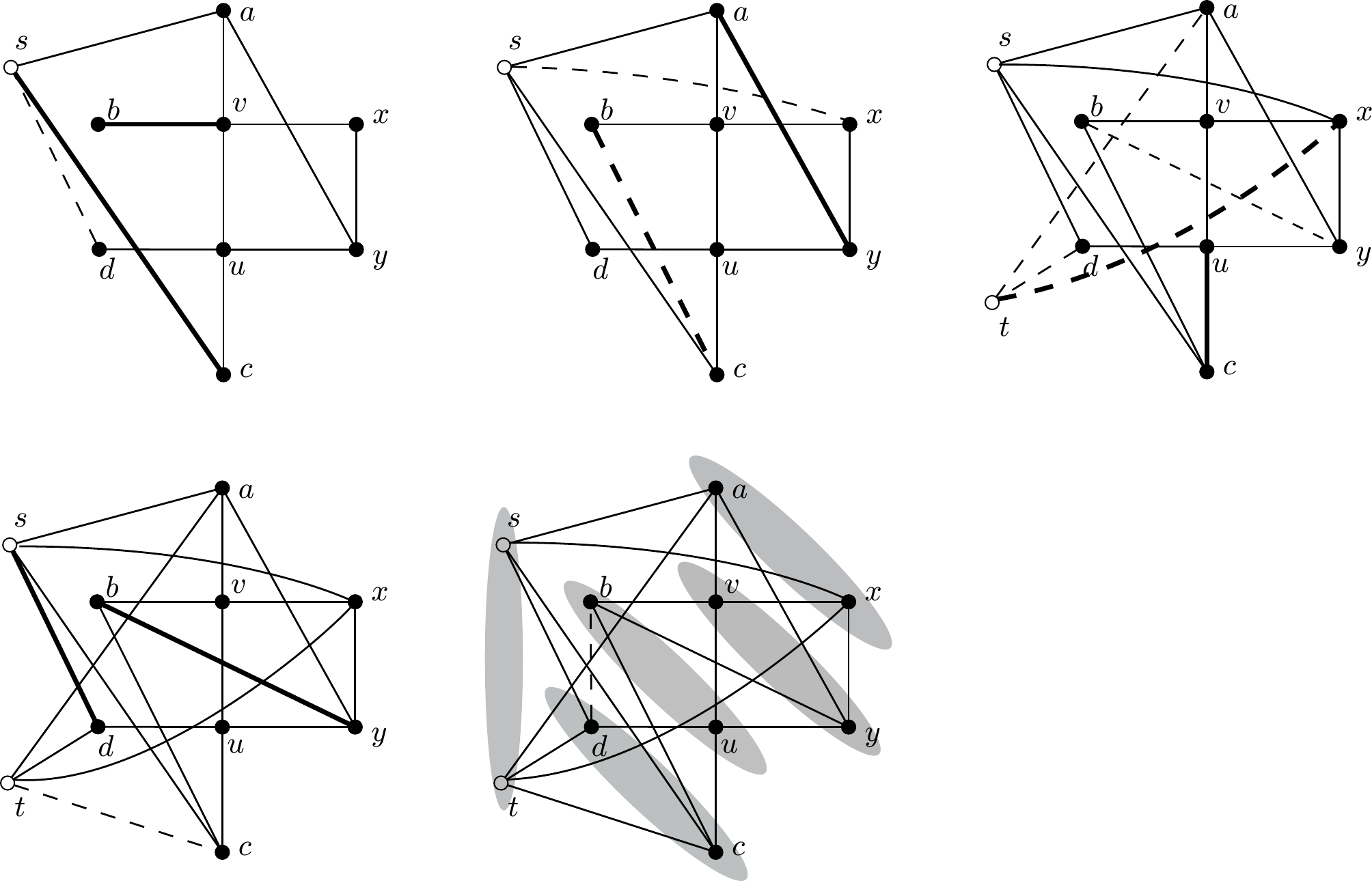}
 \caption{An illustration for Case 2.1; dashed edges indicate newly
``found'' edge and thick edges indicate considered pair in each step.}
 \label{fig:3}
\end{figure}
\medskip

\noindent \underline{Case 2.2}: {\em $bs\in E$ and $bc\notin E$}. \\
\indent \indent(See Figure \ref{fig:4} for an illustration.)

Then considering pair $\{sb,uy\}$ we derive that $yb\in E$. Now let $t\in
I(G')-\{s\}$, then we may assume without loss of generality that $ta,tc\in E$.
Considering pair $\{tc,by\}$, we have that $tb\in E$. Considering pair
$\{ta,ud\}$, we
obtain that $td\in E$. By considering pair $\{xy,tc\}$ we deduce that $xc\in
E$. Finally, by considering pair $\{td,xv\}$ we derive that $xd\in E$.
Therefore, $V(G)=X\cup I(G')$ and $G=C_{2,5}$ where the pairs corresponding to
the vertices of $C_5$ are $\{x,u\},\{v,y\},\{a,b\}\{c,d\},\{s,t\}$, a
contradiction to the assumption that $G\neq C_{2,5}$.

\begin{figure}[!ht]
 \centering
 \includegraphics[width=.9\linewidth]{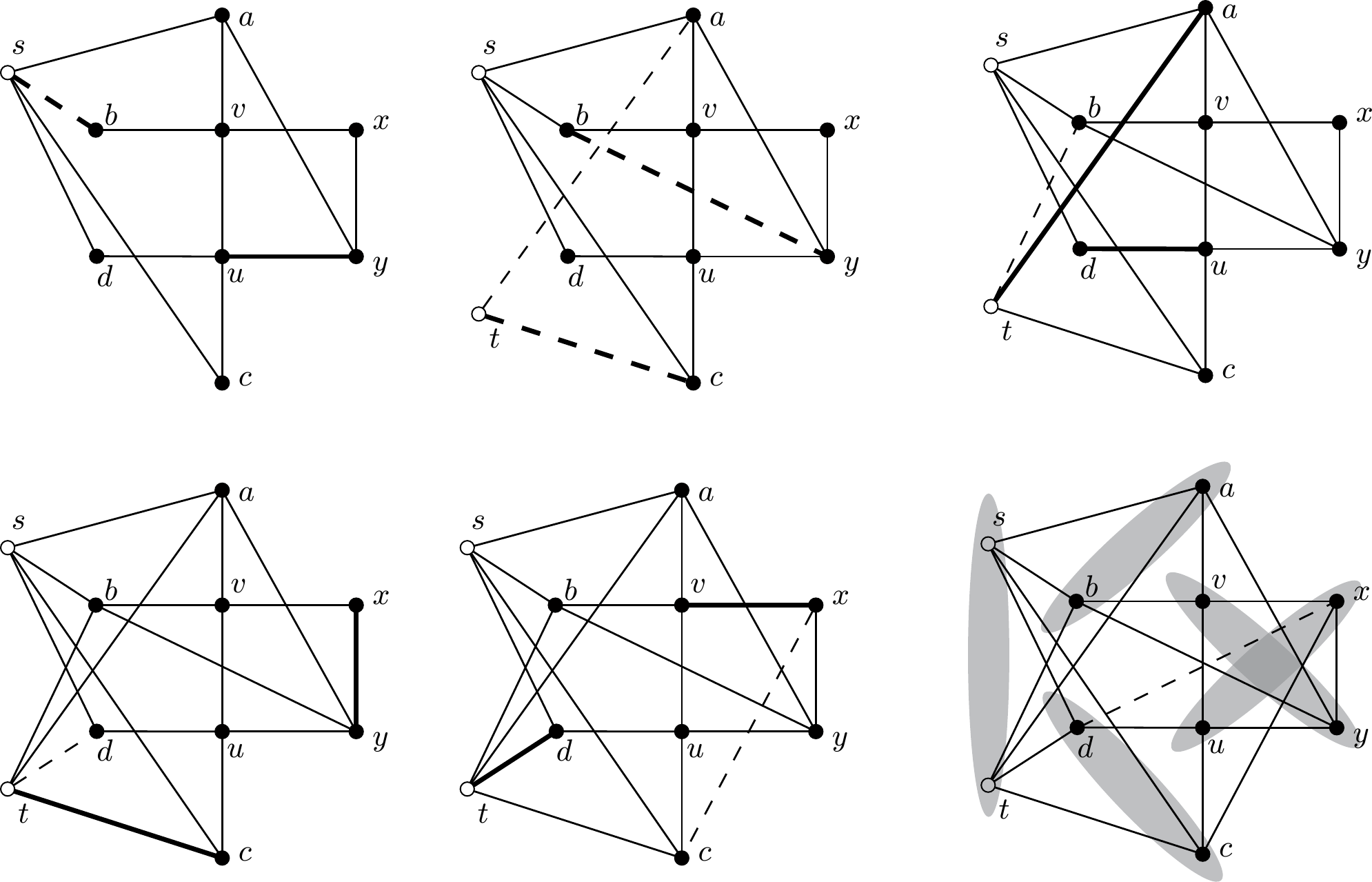}
 \caption{An illustration for Case 2.2; dashed edges indicate newly
``found'' edges and thick edges indicate considered pair in each step.}
 \label{fig:4}
\end{figure}
\end{proof}

The example in Figure \ref{fig:12}, due to Joos \cite{Joos14}, shows that the
lower bound in Theorem \ref{thm:main} is tight. It is easy to see that the
proof of Theorem \ref{thm:main} implies a polynomial time algorithm to find an
induced matching of size $\frac{1}{9}n(G)$. 
\begin{figure}[!ht]
 \centering
 \includegraphics[width=.20\linewidth]{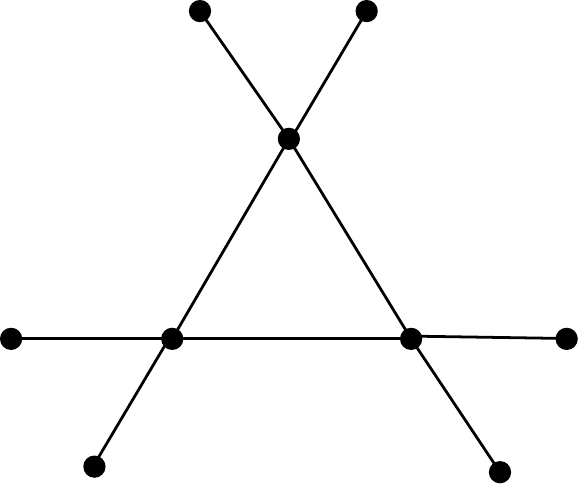}
 \caption{A tight example to Theorem \ref{thm:main}.}
 \label{fig:12}
\end{figure}


\begin{thebibliography}{}

\bibitem{Andersen92} L. D. Andersen, 
\textit{The strong chromatic index of a cubic graph is at most $10$},
Discrete Math. 108 (1992), no. 1-3, 231--252. 

\bibitem{BES07}A. Brandstädt, E. Eschen, and R. Sritharan, 
 \textit{The induced matching and chain subgraph cover problems for convex
bipartite graphs},
Theoret. Comput. Sci. 381 (2007), no. 1-3, 260--265.

\bibitem{BH08} A. Brandstädt and C. T. Ho\`{a}ng, 
\textit{Maximum induced matchings for chordal graphs in linear time},
Algorithmica 52 (2008), no. 4, 440--447.

\bibitem{Cameron89} K. Cameron, 
\textit{Induced matchings},
Discrete Appl. Math. 24 (1989), no. 1-3, 97--102.


\bibitem{Cameron04} K. Cameron, 
\textit{Induced matchings in intersection graphs.} 
Discrete Math. 278 (2004), no. 1-3, 1--9.


\bibitem{CST03}K. Cameron, R. Sritharan, and Y. Tang,
\textit{Finding a maximum induced matching in weakly chordal graphs},
Discrete Math. 266 (2003), no. 1-3, 133--142. 

\bibitem{Chang03} J.-M. Chang,
\textit{Induced matchings in asteroidal triple-free graphs},
Discrete Appl. Math. 132 (2003), no. 1-3, 67--78.

\bibitem{DDL13}K. K. Dabrowski, M. Demange, and V. V. Lozin, \textit{New results
on
maximum induced matchings in bipartite graphs and beyond},
Theoret. Comput. Sci. 478 (2013), 33--40. 

\bibitem{DWZ05} W. Duckworth, D. F. Manlove, and M. Zito, 
\textit{On the approximability of the maximum induced matching
problem},
J. Discrete Algorithms 3 (2005), no. 1, 79--91. 

\bibitem{Edmonds65}J. Edmonds, 
\textit{Paths, trees, and flowers},
Canad. J. Math. 17 (1965) 449--467. 

\bibitem{FSGT90}R. J. Faudree, R. H. Schelp, A. Gyárfás, and Zs. Tuza, 
\textit{The strong chromatic index of graphs.}
Ars Combin. 29 (1990), B, 205--211. 

\bibitem{GL93} M. C. Golumbic, and R. C. Laskar, 
 \textit{Irredundancy in circular arc graphs},
 Discrete Appl. Math. 44 (1993), no. 1-3, 79--89. 
 
\bibitem{GL00}M. C. Golumbic and M. Lewenstein, 
 \textit{New results on induced matchings},
 Discrete Appl. Math. 101 (2000), no. 1-3, 157--165.
 
 \bibitem{HQT93}P. Horák, H. Qing, and W. T. Trotter,
\textit{Induced matchings in cubic graphs},
J. Graph Theory 17 (1993), no. 2, 151--160. 

\bibitem{Joos14} F. Joos, 
\textit{Induced matchings in graphs of bounded maximum degree},
arXiv:1406.2440, June 2014.

\bibitem{JRS14}F. Joos, D. Rautenbach, and T. Sasse, 
\textit{Induced matchings in subcubic graphs}, 
SIAM J. Discrete Math. 28 (2014), no. 1, 468--473. 

\bibitem{KR03}D. Kobler and U. Rotics, 
\textit{Finding maximum induced matchings in subclasses of claw-free and
$P_5$-free graphs, and in graphs with matching and induced matching of equal
maximum size},
Algorithmica 37 (2003), no. 4, 327--346.

\bibitem{KMM12} R. J. Kang, M. Mnich, and T. Müller, 
\textit{Induced matchings in subcubic planar graphs}
SIAM J. Discrete Math. 26 (2012), no. 3, 1383--1411.

\bibitem{Lozin02} V. V. Lozin
\textit{On maximum induced matchings in bipartite graphs}, 
Inform. Process. Lett. 81 (2002), no. 1, 7--11. 


\bibitem{MR97} M. Molloy and B. Reed,
\textit{A bound on the strong chromatic index of a graph}, 
J. Combin. Theory Ser. B 69 (1997), no. 2, 103--109. 

\bibitem{SV82} L. J. Stockmeyer and V. V. Vazirani, 
\textit{NP-completeness of some generalizations of the maximum matching
problem}, 
Inform. Process. Lett. 15 (1982), no. 1, 14--19. 

\bibitem{Zito99} M. Zito, 
\textit{Induced matchings in regular graphs and
trees}, 
Lecture Notes in Comput. Sci., 1665, Springer, Berlin, 1999.


 
\end{thebibliography}
\end{document}